\documentclass[12pt,reqno,twoside]{article}

\usepackage[T1]{fontenc}

\usepackage[]{amsmath, amssymb, amsthm, tabularx}

\usepackage[]{a4, xcolor, here}

\usepackage[]{pstricks, pst-text, pst-node, pst-tree, gastex}

\usepackage{tikz}

\usepackage[tikz]{bclogo}

\usepackage{boites,graphicx}

\usepackage{soul}

\definecolor{pastelyellow}{rgb}{0.99, 0.99, 0.59}
\definecolor{aqua}{rgb}{0.0, 1.0, 1.0} 
\definecolor{aquamarine}{rgb}{0.5, 1.0, 0.83} 
\definecolor{bananayellow}{rgb}{1.0, 0.88, 0.21}
\definecolor{burgundy}{rgb}{0.5, 0.0, 0.13}
\definecolor{ao(english)}{rgb}{0.0, 0.5, 0.0}

\setul{}{0.2ex}
\setulcolor{bananayellow}

\usepackage[normalem]{ulem}

\usepackage{mathrsfs}

\usepackage{stmaryrd}

\usepackage{enumerate}

\usepackage{paralist}

\usepackage{multienum}

\usepackage{multicol}

\usepackage{fancyhdr}

\usepackage{datetime}

\usepackage{everypage}

\usepackage[contents={},opacity=1,scale=1.6,
color=gray!90]{background}

\usepackage{ifthen}

\usepackage[]{hyperref} 

\hypersetup{
		colorlinks = true,
		linkcolor = red,
		anchorcolor = black,
		citecolor = blue, 
		filecolor = cyan,
		menucolor = red,
		runcolor = cyan,
		urlcolor = cyan,
		linkbordercolor = {white},
		linktocpage = true
}

\newtheorem{theorem}{Theorem}[section]
\newtheorem{proposition}[theorem]{Proposition}
\newtheorem{lemma}[theorem]{Lemma}
\newtheorem{corollary}[theorem]{Corollary}

\theoremstyle{definition}

\newtheorem{example}[theorem]{Example}
\newtheorem{remark}[theorem]{Remark}

\makeatletter
\def\thmhead@plain#1#2#3{%
  \thmname{#1}\thmnumber{\@ifnotempty{#1}{ }\@upn{#2}}%
  \thmnote{ {\the\thm@notefont#3}}}
\let\thmhead\thmhead@plain
\makeatother

\newcommand{\bG}{\mathbf{G}}
\newcommand{\bH}{\mathbf{H}}
\newcommand{\bN}{\mathbf{N}}

\newcommand{\bT}{\mathbf{T}}

\newcommand{\cC}{\mathcal{C}}

\newcommand{\cF}{\mathcal{F}}
\newcommand{\cG}{\mathcal{G}}

\newcommand{\cL}{\mathcal{L}}

\newcommand{\cP}{\mathcal{P}}

\newcommand{\cS}{\mathcal{S}}
\newcommand{\cU}{\mathcal{U}}
\newcommand{\cV}{\mathcal{V}}

\newcommand{\rsp}{\mathrm{rowsp}}
\newcommand{\stab}{\mathrm{Stab}}
\newcommand{\orb}{\mathrm{Orb}}

\newcommand{\bbF}{{\mathbb F}}

\renewcommand{\leq}{\leqslant}

\newcommand{\car}{\mathrm{char}}
\newcommand{\mat}{\mathrm{Mat}}

{\end{list}}%

\begin{document}

\renewcommand{\headrulewidth}{0pt}

\rhead{ }
\chead{\scriptsize An Orbital Construction of Optimum Distance Flag Codes}
\lhead{ }

\title{An Orbital Construction of\\ Optimum Distance Flag Codes}

\author{\renewcommand\thefootnote{\arabic{footnote}}
Clementa Alonso-Gonz\'alez\footnotemark[1],\,  Miguel \'Angel Navarro-P\'erez\footnotemark[1], \\ 
\renewcommand\thefootnote{\arabic{footnote}} 
 Xaro Soler-Escriv\`a\footnotemark[1]}

\footnotetext[1]{Dpt.\ de Matem\`atiques, Universitat d'Alacant, Sant Vicent del Raspeig, Ap.\ Correus 99, E -- 03080 Alacant. \\ E-mail adresses: \texttt{clementa.alonso@ua.es, miguelangel.np@ua.es, xaro.soler@ua.es.}}

{\small \date{\today}} 

\maketitle

\begin{abstract}
Flag codes are multishot network codes consisting of sequences of nested subspaces (flags) of a vector space $\bbF_q^n$, where $q$ is a prime power and $\bbF_q$, the finite field of size $q$. In this paper we study the construction on $\bbF_q^{2k}$ of full flag codes having maximum distance (\emph{optimum distance full flag codes}) that can be endowed with an orbital structure provided by the action of a subgroup of the general linear group. More precisely,  starting from a subspace code of dimension $k$ and maximum distance with a given orbital description, we provide sufficient conditions to get an optimum distance full flag code on $\bbF_q^{2k}$ having an orbital structure directly induced by the previous one. In particular, we exhibit a specific orbital construction with the best possible size from an orbital construction of a planar spread on $\bbF_q^{2k}$ that strongly depends on the characteristic of the field.
\end{abstract}



\section{Introduction}\label{sec:Introduction}

Random linear network coding is a powerful tool for data transmission over noisy and lossy networks. First introduced in \cite{HoKoetMedKarEff03}, this network operates with {\em packets}, consisting of $n$ symbols of a finite field $\bbF_q$ (for a prime power $q$), which can be interpreted as vectors in the vector space $\bbF_q^n$. The intermediate nodes of the network transmit random linear combinations of the packets they receive.  The seminal paper in this area is \cite{KoetKschi08}, where K\"otter and Kschischang realized that the vector subspace generated by the packets sent by the source is preserved in the network. Thus, they proposed to encode the information in vector subspaces of $\bbF_q^n$ and defined the {\em subspace channel}.  In this setting, a {\em subspace code} is just a non-empy subset of the {\em projective space} $\cP_q(n)$, which is the collection of all vector subspaces of $\bbF_q^n$. Since then, the study of these codes has led to many papers in recent years (see for instance \cite{Gref2018} and the references inside). Most of this research focuses on the study of {\em constant dimension codes}, that is, subspace codes in which all subspaces have the same dimension (we refer the reader to \cite{TrautRosen18} for a survey on constructions of these codes). A leading family of constant dimension codes are {\em spread codes} (or just {\em spreads}), since they attain the maximum possible distance and the largest size for that distance. A spread of $\bbF_q^n$ is a collection of subspaces of $\bbF_q^n$, all of the same dimension, which partition the ambient space. Such a family of subspaces exists if, and only if, their dimension divides $n$ \cite{Segre64}. Within the scope of network coding, spread codes were introduced in \cite{ManGorRos2008}. Another interesting family of constant dimension codes is the one of {\em orbit codes}. First defined in \cite{TrautManRos2010}, these codes appear as orbits of subspaces under the action of a subgroup of the general linear group, $GL(n,q)$. Many of the known families of constant dimension codes are, in fact, orbit codes \cite{TrautManRos2010}. Due to their orbital structure, the size of these codes clearly depends on the order of the acting group and their minimum distance can be computed in a simpler manner \cite{TrautManBraRos2013}. 

Subspace codes are also known as {\em one-shot subspace codes} since they use the subspace channel only once. In contrast, when the subspace channel is used many times, we talk about  {\em multishot subspace codes} \cite{NobUcho09}. When the parameters $q$ and $n$ cannot be increased, multishot subspace codes arise as an interesting substitute of the standard subspace codes in order to improve efficiency. In this context, an {\em $r$-shot subspace code} over $\cP_q(n)$ is a non-empty subset of $\cP_q(n)^r$ and a {\em codeword} is just an $r$-tuple of subspaces of $\cP_q(n)$, which is sent through the subspace channel using $r$ shots. 

In this paper, we focus on a special type of multishot subspace codes, in which the codewords are sequences of nested subspaces in $\cP_q(n)$, that is, {\em flags}. The use of flags in network coding was introduced in \cite{LiebNebeVaz2018}, where the authors developed a network channel model for flags and  provided constructions of flag codes using group actions. Later on, flag codes having maximum distance, called {\em optimum distance flag codes}, have been studied in \cite{CasoPlanar,CasoGeneral}. In  \cite{CasoPlanar},  these flag codes are characterized in terms of the constant dimension codes used at each shot and, for $n=2k$, it is proved that optimum distance full flag codes with the largest size are exactly those having a spread at the $k$-th shot. In the same paper, the authors provide a precise construction of these codes as well as a decoding algorithm.

The main objective of the present paper is to provide an orbital construction of optimum distance full flag codes on $\bbF_q^{2k}$ with the best size. Taking into account the results obtained in \cite{CasoPlanar}, we want to do this by extending to flags the orbital expression appearing in \cite{TrautManRos2010} of a specific planar spread $\cS$ under the action of a particular group $\bG$. In a general setting, given a subspace code of dimension $k$ and maximum distance equipped with an orbital structure provided by the action of a given subgroup of $GL(2k,q)$, we investigate how this action can be extended to the full flag variety on $\bbF_q^{2k}$ in such a way we obtain optimum distance full flag codes with an orbital structure directly inherited from the one fixed at dimension $k$. In particular, we conclude that the action of the chosen group $\bG$ on the planar spread $\cS$ cannot be properly extended to the flag codes level. Nevertheless, we find a suitable subgroup $\bH$ of $\bG$ whose action on $\cS$ also works for full flags on $\bbF_{q}^{2k}$. This action, in addition, allows us to construct optimum distance full flag codes with the largest size and a number of orbits that deeply depends on the parity of the characteristic of the field $\bbF_q$.

The article is organized as follows. In Section \ref{sec:Preliminaries} we summarize all the basic background we need on finite fields and subspace codes. Section \ref{sec:flags} begins  with a brief review of known results on flag codes, which leads to the study of some properties of orbit flag codes. We give sufficient conditions to obtain optimum distance orbit full flag codes on $\bbF_q^{2k}$ from orbit codes of dimension $k$ with maximum distance. Moreover, we show that, under these conditions, such codes codes directly inherit the orbital structure fixed at dimension $k$. Section \ref{sec:construc} is devoted to present a specific orbital construction of optimum distance full flag codes on $\bbF_q^{2k}$ with the best size. Taking $\cS$ and $\bG$ as in the preceding paragraph, we find a suitable subgroup $\bH$ of  $\bG$, whose action on $\cS$ can be extended to flags in the appropriate manner. In particular, by translation of the problem to the spread of lines of $\bbF_{q^k}^2$, we show that this group $\bH$ provides a nice orbital description of the spread $\cS$ strongly depending on the characteristic of the field. This orbital description of $\cS$ naturally induces an orbital construction of optimum distance full flag codes with the best size as union of at most two orbits.


\section{Preliminaries}\label{sec:Preliminaries}

In this section we collect some definitions and known results we will need later on. First, we give some background on finite fields; then we briefly review some concepts and facts on subspace codes and two special families of them: spread codes and orbit codes.

\subsection{Finite fields}\label{subsec:finitefields}
We recall some known facts about finite fields that can be found in \cite{Niede}. 

Let $q$ be a prime power. We denote the finite field of $q$ elements by $\bbF_q$. 
Given a positive integer $k$, a generator $\alpha$ of the multiplicative cyclic group $\bbF_{q^k}^*$ is said to be a {\em primitive element} of the field $\bbF_{q^k}$. In this case, the multiplicative order of $\alpha$ is $q^k-1$. Moreover, if $p(x)=x^k+\sum_{i=0}^{k-1}p_ix^i$ is the minimal polynomial of $\alpha$ over $\bbF_q$, we say that $p(x)$ is a {\em primitive polynomial} and we have the following isomorphisms of fields
\[
\bbF_{q^k}\cong \bbF_q[x]/(p(x)) \cong \bbF_{q}[\alpha].
\]
It turns out that the finite field $\bbF_{q^k}$ can be realized as $\bbF_{q^k}\cong \bbF_{q}[\alpha]=\{0\}\cup \langle \alpha\rangle$.

On the other hand, the {\em companion matrix} $P$ of $p(x)$ over $\bbF_q$ is defined as the following $k\times k$ matrix 
\[
P=
\begin{pmatrix}
 0 & 1        &  0        &  \cdots & 0\\
  0      & 0  &  1        &  \cdots & 0\\
  \vdots &  \vdots  &   \vdots  &  \ddots & \vdots \\
   0      & 0        &  0  &  \cdots & 1\\
  -p_0      & -p_1        &  -p_2        &  \cdots & -p_{k-1}
\end{pmatrix}.
\]
The characteristic polynomial of $P$ is $p(x)$. Therefore, the eigenvalues of $P$ are precisely the zeros of $p(x)$, that is, the elements $\alpha, \alpha^q, \ldots, \alpha^{q^{k-1}}$. In particular, it follows that $\det(P)=\alpha^{1+q+q^2+\dots +q^{k-1}}=\alpha^t$, where $t=\frac{q^k-1}{q-1}$. Moreover, the $\bbF_q$-algebra $ \bbF_{q}\left[ P \right]$ is isomorphic to the finite field $\bbF_{q^k}$, since the map  
\begin{equation}\label{eq:isoalphaP}
\begin{array}{cccc}
\phi : & \bbF_{q}[\alpha]                 &  \longrightarrow &         \bbF_{q}\left[ P \right]\\
         & \sum_{i=0}^{k-1}a_i\alpha^i       & \longmapsto & \sum_{i=0}^{k-1}a_iP^i
\end{array}
\end{equation}
is a field isomorphism mapping $\alpha$ into $P$. As a consequence, $P$ is a primitive element of the field $\bbF_q[P]$, that is, $P$ has multiplicative order $q^k-1$ and $\bbF_q[P]=\{0\}\cup \langle P\rangle$.

\subsection{Subspace codes, spreads and orbit codes}

Given the finite field $\bbF_q$ and a positive integer $n$, the set of all vector subspaces of $\bbF_q^n$,  denoted by  $\cP_q(n)$, is a metric space with respect to the subspace distance defined by 
\begin{equation}\label{eq:d_S}
d_S(\cU,\cV)=\dim(\cU)+\dim(\cV)-2\dim(\cU\cap\cV)
\end{equation}
for all $\cU,\cV\in \cP_q(n)$. 
For any integer $k\in\{0,\ldots, n\}$, the set of all $k$-dimensional vector subspaces of $\bbF_q^n$ is called the {\em Grassmann variety} (or simply the {\em Grassmannian}) of dimension $k$ and it is denoted by $\cG_q(k,n)$. Thus, it follows that $\cP_q(n)=\cup_{k=0}^n \cG_q(k,n)$.

A {\em subspace code} is a non-empty set $\cC$ of $\cP_q(n)$ and the {\em minimum (subspace) distance} of $\cC$ is defined as 
\[
d_S(\cC)=\min\{d_S(\cU,\cV)\ |\ \cU,\cV\in\cC, \cU\neq \cV\}.
\]
Notice that $d_S(\cC)>0$ whenever $|\cC|>1$. If $|\cC|=1$, we put $d_S(\cC)=0$. A subspace code $\cC$ is said to be a {\em constant dimension code} if $\cC\subseteq \cG_q(k,n)$,  for some dimension $k\in\{1,\ldots, n-1\}$. The minimum distance of such a code is an even integer and it is upper-bounded by  $d_S(\cC)\leq \min\{2k, 2(n-k)\}$. This bound is attained by codes in which the intersection of every two different subspaces has dimension equal to $\max\{ 0, 2k-n\}$ (see \cite{KoetKschi08} for all the basics on subspace codes and constant dimension codes). Now, we introduce the concepts of spread and orbit codes, which are the families of constant dimension codes that we need in this paper.


\subsubsection{Spreads}\label{subsec:spreads}

A $k$-dimensional {\em spread} of $\bbF_q^n$, or a {\em $k$-spread}, is a set of elements of $\cG_q(k,n)$ that pairwise intersect trivially and cover the whole space $\bbF_q^n$. Spreads are geometrical objects (see \cite{Hirschfeld98, Segre64}) that can be used as constant dimension codes. In this case, we speak of {\em spread codes}. Spreads in $\cG_q(k,n)$ exist if, and only if, $k$ divides $n$.  One can easily compute the minimum distance of a $k$-spread code of $\bbF_q^n$ as  $2k$ and its cardinality as $(q^n-1)/(q^k-1)$  (see \cite{Hirschfeld98, ManGorRos2008}). It follows that spreads are an interesting type of constant dimension codes since they have the maximum possible distance and the best possible size for that distance. We will focus on {\em planar spreads}, that is, spreads where the sum of any two different elements covers the whole space $\bbF_q^n$. Equivalently, a planar spread is a $k$-spread of $\bbF_q^{2k}$.

We will use a concrete construction of planar spreads of $\bbF_q^{2k}$ inspired in the general construction of spreads due to Segre \cite{Segre64} and first used in the network coding setting in \cite{ManGorRos2008}. To do this, we use the matrix representation of a $k$-dimensional subspace $\cU\in  \cG_q(k,n)$. A {\em generator matrix} $U\in\mat_{k\times n}(\bbF_q)$ of $\cU$, is a full-rank matrix whose rows form a basis of $\cU$. Hence, 
$
\cU=\rsp(U)=\{xU\ |\ x\in\bbF_q^k\}
$.

Let  $P$ be the companion matrix of a primitive polynomial $p(x)\in \bbF_q[x]$ of degree $k$. If we denote by $0_k, I_k\in\mat_{k\times k}(\bbF_q)$ respectively to the $k\times k$ zero and  identity matrices, then the collection of subspaces 
\begin{equation}\label{def: spread Manganiello}
{\cal S}= \left\{ \rsp (0_k | I_k ), \rsp ( I_k | 0_k ), \rsp ( I_k | P^i )\ |\   i=0, \ldots, q^k-2  \right\}\end{equation}
is a planar spread code of $\bbF_q^{2k}$ (see \cite{ManGorRos2008}). In the same paper, the authors gave another way to view this spread by using the field isomorphism $\phi$ given in (\ref{eq:isoalphaP}). Notice that the following map, which is called {\em field reduction map}, is a natural embedding of the Grassmannian of lines of $\bbF_{q^k}^2$ into  $\cG_{q}(k,2k)$:
\begin{equation}\label{eq:fieldreduction}
\begin{array}{cccc}
\varphi: & \cG_{q^k}(1,2)   & \longrightarrow        & \cG_{q}(k,2k)\\
             & \rsp(x_1, x_2)   & \longmapsto & \rsp(\phi(x_1) | \phi(x_2)).
\end{array}
\end{equation}
Clearly, $\cG_{q^k}(1,2)$ is a planar spread of $\bbF_{q^k}^2$ and, as proved in \cite{ManGorRos2008}, $\varphi(\cG_{q^k}(1,2))=\cS$ is a planar spread of $\bbF_{q}^{2k}$. We will make use of this bijective map between $\cG_{q^k}(1,2)$ and $\cS$ in the following sections. 

\subsubsection{Orbit codes}\label{subsec:orbitcodes}

Constant dimension codes arising from group actions are called {\em orbit codes} and were introduced in \cite{TrautManRos2010}. In order to define them, we consider the general linear group of degree $n$ over $\bbF_q$, denoted by $GL(n,q)$, composed of all invertible matrices of $\mat_{n\times n}(\bbF_q)$. Now, given a full-rank matrix $U\in\mat_{k\times n}(\bbF_q)$ and the subspace $\cU=\rsp(U)\in \cG_q(k,n)$, one has the following group action on the Grassmann variety:
\[
\begin{array}{ccc}
\cG_q(k,n) \times GL(n,q) & \longrightarrow & \cG_q(k,n)\\
(\cU, A)                            & \longmapsto & \cU\cdot A=\rsp(UA).
\end{array}
\]
This operation is well-defined since it is independent from the generator matrix we choose for $\cU$ (see \cite{TrautManRos2010}).

This action allows to construct constant dimension codes starting from a given subspace. To be precise, if $\cU$ is a $k$-dimensional subspace of $\bbF_q^n$, its orbit under the action of some subgroup $\bT$ on $\cG_q(k, n)$ is a constant dimension code that we denote by
\[
\orb_{\bT}(\cU)=\{\cU\cdot A\ |\ A\in \bT \}\subseteq \cG_q(k,n) .
\]
Codes constructed in this way are called {\em orbit codes}. The  stabilizer of $\cU$ under the action of $\bT$ is the subgroup $\stab_{\bT}(\cU)=\{A\in\bT\ |\ \cU\cdot A=\cU\}$ and the size of the orbit code $\orb_{\bT}(\cU)$ is
\[
|\orb_{\bT}(\cU)|=\frac{|\bT|}{|\stab_{\bT}(\cU)|}.
\]
Moreover, the minimum distance of $\orb_{\bT}(\cU)$ can be computed as
\[
d_S(\orb_{\bT}(\cU))=\min\{d_S(\cU,\cU\cdot A)\ |\ A\in \bT\setminus \stab_{\bT}(\cU)\}.
\]
The reader is referred to \cite{TrautManRos2010, TrautManBraRos2013} for a general background on  orbit codes.

In \cite{TrautManRos2010} the authors suggest an orbital description of the planar spread $\cal S$ provided in (\ref{def: spread Manganiello}). Let us recall it. For the rest of the paper, we fix ${\cal U}_k=\rsp( I_k | 0_k )$, that is, the standard $k$-dimensional vector subspace of $\bbF_q^{2k}$. If $P$ is the companion matrix considered in  (\ref{def: spread Manganiello}), then

\begin{equation}\label{def: spread Manganiello orbital}
{\cal S}= \orb_{\bG}({\cal U}_k) = \{ {\cal U}_k\cdot A\ |\ A\in \bG\} \subseteq \cG_q(k,2k),
\end{equation}
where $\bG$ is  the following subgroup of $GL(2k,q)$
  
\begin{equation}\label{def: Grupo grande}
\bG=
\left\langle
\begin{pmatrix}
0_k & I_k\\
I_k & 0_k
\end{pmatrix}, \ 
\begin{pmatrix}
I_k& P^i\\
0_k & I_k
\end{pmatrix} \ \Big | \ i=0, \ldots, q^k-2
\right\rangle.
\end{equation}

This orbital construction of $\cS$, as the orbit of $\cU_k$ under the action of the group $\bG$, will be crucial in this paper. 


\section{Flag codes}\label{sec:flags}

Flag varieties on the vector space $\bbF_q^n$  are classical geometrical concepts that can be seen as an  extension of Grassmann varieties. In the setting of network coding, the use of flags was introduced in \cite{LiebNebeVaz2018}. We start the section by reviewing some known results on flag codes. Next, we will focus on studying some properties of orbit flag codes. From there, we will give sufficient conditions to build orbit full flag codes on $\bbF_q^{2k}$ with the best possible distance from orbit codes in $\cG(k,2k)$ attaining the maximum distance. Finally, we show how to increase the cardinality of these flag codes by considering the union of orbits under the action of the same group.   

\subsection{Background on flag codes}
 
A {\em flag} of type $(t_1, \ldots, t_r)$ on  $\mathbb{F}_q^n$, with $0<t_1<\cdots <t_r<n$,  is an element $\mathcal{F}=(\mathcal{F}_1,\ldots,  \mathcal{F}_r)$ of $\mathcal{G}_q(t_1,n) \times \cdots \times \mathcal{G}_q(t_r,n) \subseteq \mathcal{P}_q(n)^r$ such that
\[
\{0\}\subsetneq \mathcal{F}_1 \subsetneq \cdots \subsetneq \mathcal{F}_r \subsetneq \mathbb{F}_q^n.
\]
We say that $\mathcal{F}_i$ is the {\em $i$-th subspace} of the flag $\cF$. When the type vector is $(1, 2, \ldots, n-1),$ we speak about {\em full flags}.  The {\em flag variety} of type $(t_1, \ldots, t_r)$ on $\mathbb{F}_q^n$ is the space of flags of type $(t_1, \ldots, t_r)$ on $\mathbb{F}_q^n$ and will be denoted by $\mathcal{F}_q((t_1,\ldots, t_r),n)$. It can be naturally endowed with a metric by extending the subspace distance provided in (\ref{eq:d_S}). 
Given  $\cF=(\mathcal{F}_1,\ldots,  \mathcal{F}_r)$ and $\cF'=(\mathcal{F}'_1,\ldots,  \mathcal{F}'_r)$ two flags in $\mathcal{F}_q( (t_1, \ldots, t_r),n)$, the {\em flag distance} between $\cF$ and $\cF'$ is
\[
d_f(\cF,\cF')= \sum_{i=1}^r d_S(\mathcal{F}_i, \mathcal{F}'_i).
\] 

A {\em flag code of type $(t_1,\ldots,t_r)$} on $\bbF_q^n$ is a non-empty subset ${\cal C}$ of $ \mathcal{F}_q((t_1,\ldots, t_r),n)$. The {\em minimum distance} of the flag code $\cC$ is given by
\[
d_f(\cC)=\min\{d_f(\cF,\cF')\ |\ \cF,\cF'\in\cC, \ \cF\neq \cF'\}.
\]
If $|\cC|=1$, we put $d_f(\cC)=0$. Notice that $d_f(\cC)$ is upper bounded by (see \cite{CasoPlanar})
\begin{equation}\label{eq:quotamaxdistflag}
d_f(\mathcal{C}) \leq 2\left(\sum_{t_i \leq \lfloor \frac{n}{2}\rfloor} t_i + \sum_{t_i > \lfloor \frac{n}{2}\rfloor} (n-t_i)\right).
\end{equation}

Flag codes attaining the bound (\ref{eq:quotamaxdistflag}) are called {\em optimum distance flag codes} and have been studied in  \cite{CasoPlanar, CasoGeneral}. The {\em $i$-projected code} $\cC_i$ of $\cC$ is the constant dimension code given by the set of $i$-th subspaces of flags in $\cC$, that is, 
\[
\mathcal{C}_i= \left\lbrace \mathcal{F}_i  \ | \ (\cF_1,\ldots, \cF_i, \dots, \cF_r) \in \mathcal{C}\right\rbrace \subseteq \mathcal{G}_q(t_i,n),
\]
 for $i=1,\ldots, r$ (see \cite{CasoPlanar}). Notice that $1\leq |\cC_i|\leq |\cC|$ for all $i=1,\ldots, r$. 
 
Having projected codes with the maximum possible distance is a necessary condition in order to obtain optimum distance flag codes, but it is not enough as it was proved in \cite{CasoPlanar}. For this reason, in that paper, the concept of {\em disjointness} was introduced. The flag code $\cC$ is said to be {\em disjoint} if  $\vert \mathcal{C}_1\vert =\cdots=\vert\mathcal{C}_r\vert=\vert\mathcal{C}\vert$. Using this definition, one has the following characterization of optimum distance flag codes.

\begin{theorem} \cite[Th. 3.11]{CasoPlanar}\label{teo:carac_odfc}
Let $\cC$ be a flag code of type $(t_1, \ldots, t_r)$. The following statements are equivalent:
\begin{enumerate}
\item[(i)] The code $\cC$ is an optimum distance flag code.
\item[(ii)] The code $\cC$ is disjoint and every projected code $\cC_i$ attains the maximum possible distance.
\end{enumerate}
\end{theorem}

In  \cite{CasoPlanar} it is proved that,  for $n=2k$,  optimum distance full flag codes with the largest size are exactly those having a planar spread as a projected code.

\begin{theorem}\cite[Cor. 4.3]{CasoPlanar}\label{teo:odfc_maxsize}
Let $\cC$ be an optimum distance full flag code on $\bbF_q^{2k}$. Then $|\cC|\leq q^k+1$. Equality holds if, and only if, the $k$-projected code of $\cC$ is a planar spread of $\bbF_q^{2k}$.
\end{theorem}

\subsection{Orbit flag codes}\label{sec:propFlagsOrb}

Since flags are sequences of subspaces of $\bbF_q^n$, the group action given in Section \ref{subsec:orbitcodes} on the Grassmannian can be easily extended to the space $\cF_q((t_1, \ldots, t_r),n)$. Hence, given a flag $\cF=(\cF_1, \ldots, \cF_r)$ of type $(t_1, \ldots, t_r)$ on $\bbF_q^n$ and a subgroup $\bT$ of $GL(n,q)$, the {\em orbit flag code} generated by $\cF$ under the action of $\bT$ is
$$
\orb_\bT(\cF)=\{\cF\cdot A\ |\ A\in\bT\}=\{(\cF_1\cdot A, \ldots, \cF_r\cdot A)\ |\ A\in\bT\}\subseteq \cF_q((t_1, \ldots, t_r),n) .
$$
It follows that the size of an orbit flag code can be computed as 
\[|\orb_\bT(\cF)|=\frac{|\bT|}{|\stab_\bT(\cF)|},
\]
where $\stab_\bT(\cF)=\{A\in\bT\ | \  \cF\cdot A=\cF\}$  denotes the stabilizer subgroup of the flag $\cF$ under the action of $\bT$ and its minimum distance is 
\[
d_f(\orb_\bT(\cF))=\min\{d_f(\cF,\cF\cdot A)\ |\ A\in \bT\setminus \stab_{\bT}(\cF)\}.
\]

In this context, the projected codes of the orbit flag code $\orb_\bT(\cF)$ are orbit subspace codes and can be written as
$$
\orb_\bT(\cF)_i=\orb_\bT(\cF_i)\subseteq \cG_q(t_i,n), \ i=1, \ldots, r.
$$

The next result states the relationship between the stabilizer subgroup of a flag and the ones of its subspaces under the action of the same group. The proof is straightforward.

\begin{lemma}\label{lemma: estabilizador flag}
With the previous notation, it holds
$$
\stab_\bT(\cF)=\bigcap_{i=1}^r \stab_\bT(\cF_i).
$$
\end{lemma}

Remember that we say that a flag code is disjoint if its cardinality coincides with the ones of its projected codes. Since in the orbital scenario the cardinality of a code is extremely related with the order of the stabilizer subgroup, the previous lemma allows us to establish a stronger equivalence for disjointness. The proof is clear and we omit it. 

\begin{proposition}\label{prop:disjorb}
Given a flag $\cF=(\cF_1, \ldots, \cF_r)$ of type $(t_1, \ldots, t_r)$ on $\bbF_q^n$ and a
subgroup $\bT$ of $GL(n,q)$, the following statements are equivalent:
\begin{enumerate}[a)]
    \item $\orb_\bT(\cF)$ is a disjoint flag code. 
    \item $\stab_\bT(\cF)=\stab_\bT(\cF_1)=\cdots= \stab_\bT(\cF_r)$. 
     \item $\stab_\bT(\cF_1)=\cdots= \stab_\bT(\cF_r)$. 
\end{enumerate}
\end{proposition}

As stated in Section \ref{subsec:orbitcodes}, the orbit of $\cU_k=\rsp(I_k | 0_k)$ under the action of the group $\bG$ defined in (\ref{def: Grupo grande}) provides the planar spread $\cS$ of $\bbF_q^{2k}$ (see (\ref{def: spread Manganiello orbital})). Our intention is to study the extension of this action to full flags on $\bbF_q^{2k}$ in order to construct optimum distance orbit full flag codes with $\cS$ as a projected code. Notice that, by means of Theorem \ref{teo:odfc_maxsize}, such codes would attain the maximum possible size as well. 
On the other hand, optimum distance flag codes are, in particular, disjoint (see Theorem \ref{teo:carac_odfc}). Hence, according to Proposition \ref{prop:disjorb}, we would need $\bG$ to provide equal stabilizer subgroups when acting on every subspace of a suitable flag.  Taking into account that the spread $\cS$ is a subspace code with maximum distance, the next result gives us a necessary condition for our purpose valid for any subgroup of $GL(2k,q)$.

\begin{proposition}\label{prop:contenido}
Let $\cF=(\cF_1, \ldots, \cF_{2k-1})$ be a full flag on $\bbF_q^{2k}$ and $\bT$
a subgroup of $GL(2k,q)$ such that $\orb_{\bT}(\cF_k)$ has maximum distance. Then 
$\stab_{\bT}(\cF_i)\subseteq \stab_{\bT}(\cF_k)$, for all $i\in\{1,\ldots, 2k-1\}$.
\end{proposition}
\begin{proof}
Consider $A\in\bT$ such that $\cF_i = \cF_i \cdot A$ for some $i\in\{1,\ldots, 2k-1\}$. Since $\cF_k,\,   \cF_k \cdot A\in\orb_{\bT}(\cF_k)$, we know that  $d_S(\cF_k, \cF_k \cdot A)=2k$, whenever $\cF_k\neq \cF_k \cdot A$. Then,  either $\cF_k\cap \cF_k \cdot A=\{0\}$ or $\cF_k=\cF_k\cdot A$. Hence,
 \begin{itemize}
    \item if $1\leq i\leq k$, we obtain that $\{0\}\neq \cF_i = \cF_i \cap\cF_i \cdot A\subseteq \cF_k\cap \cF_k \cdot A$. Thus, $\cF_k=\cF_k \cdot A$.    
    \item If $k < i \leq 2k-1$,  the $i$-dimensional subspace $\cF_i = \cF_i \cdot A$ contains both  $\cF_k$ and $\cF_k \cdot A$ and hence, $\cF_k=\cF_k \cdot A$. Otherwise, $\cF_i $ would contain the whole space $\bbF_q^{2k}=\cF_k\oplus\cF_k \cdot A$, which is not possible. 
\end{itemize}
Therefore,  $\stab_{\bT}(\cF_i) \subseteq \stab_{\bT}(\cF_k)$ for all $i=1, \ldots, 2k-1$.
\end{proof}

As a consequence of the previous results, we can enunciate a sufficient condition to build optimum distance orbit full flag codes on $\bbF_q^{2k}$ from constant dimension codes in $\cG(k,2k)$ attaining the maximum distance.

\begin{proposition}\label{teo:ODFCorbital}
Let $\cF=(\cF_1, \ldots, \cF_{2k-1})$ be a full flag on $\bbF_q^{2k}$ and $\bT$
a subgroup of $GL(2k,q)$ such that $\orb_{\bT}(\cF_k)$ has maximum distance. 

If $\stab_{\bT}(\cF_k)\subseteq \stab_{\bT}(\cF_i)$, for all $i\in\{1,\ldots, 2k-1\}$, then $\orb_{\bT}(\cF)$ is an optimum distance full flag code of cardinality $|\orb_{\mathbf{T}}(\cF_k)|$. 
\end{proposition}
\begin{proof}
If $\stab_{\bT}(\cF_k)\subseteq \stab_{\bT}(\cF_i)$, for all $i\in\{1,\ldots, 2k-1\}$, then $\bT$ provides equal stabilizers when acting on all subspaces of $\cF$, by Proposition \ref{prop:contenido}. Hence, by means of Proposition \ref{prop:disjorb}, we obtain that $\orb_{\bT}(\cF)$ is a disjoint full flag code.

Now, we will show that $\orb_{\bT}(\cF)_i=\orb_{\bT}(\cF_i)$ has the maximum possible distance, that is, 
\[
d_S(\orb_{\bT}(\cF_i)) =
\left\lbrace
\begin{array}{lll}
2i & \text{if} & i\leq k,\\
2(2k-i) & \text{if} & i > k.
\end{array}
\right.
\]
Arguing by contradiction, suppose that for some dimension $1\leq i\leq 2k-1$, the maximum possible distance is not attained. Hence, there must exist a matrix $A \in \bT \setminus \stab_{\bT}(\cF_i)$ such that 
\[
\left\lbrace
\begin{array}{lll}
\dim(\cF_i \cap \cF_i  \cdot A) > 0 & \text{if} & i\leq k,\\
\dim(\cF_i \cap \cF_i  \cdot A) > 2(i-k) & \text{if} & i > k.
\end{array}
\right.
\]
In any case, $A$ does not stabilize $\cF_k$,  since $\stab_{\bT}(\cF_k)= \stab_{\bT}(\cF_i)$. Thus, $d_S(\cF_k, \cF_k \cdot A)=d_S(\orb_{\bT}(\cF_k))=2k$. This implies that 
 $\cF_k\cap\cF_k \cdot A =\{0\}$ and $\cF_k\oplus\cF_k \cdot A=\bbF_q^{2k}$. Therefore,
\begin{itemize}
    \item If $i\in \{1,\ldots k\}$, then $ 0 < \dim(\cF_i \cap \cF_i \cdot A) \leq \dim(\cF_k\cap\cF_k \cdot A)=0$.
    \item If $i\in \{k,\ldots 2k-1\}$, then $\bbF_q^{2k}=\cF_k\oplus\cF_k \cdot A\subseteq \cF_i+\cF_i \cdot A$ and thus  $\dim(\cF_i \cap \cF_i  \cdot A)=2(i-k)$.
\end{itemize}
Both cases lead us to a contradiction. We conclude that  $\orb_{\bT}(\cF_i)$ must have the maximum possible distance for all dimensions $i=1,\ldots, 2k-1,$ as stated. 

Consequently, $\orb_{\bT}(\cF)$ is an optimum distance full flag code by Theorem \ref{teo:carac_odfc}.
\end{proof}

The following result is a particular case of the previous proposition. 

\begin{corollary}\label{cor:ODFCorbital}
Let $\cF=(\cF_1, \ldots, \cF_{2k-1})$ be a full flag on $\bbF_q^{2k}$ and $\bT$
a subgroup of $GL(2k,q)$ such that $\orb_{\bT}(\cF_k)$ has maximum distance. 

If $\stab_{\bT}(\cF_k)=\{I_{2k}\}$, then $\orb_{\bT}(\cF)$ is an optimum distance full flag code. 
\end{corollary}

From Proposition \ref{teo:ODFCorbital}, we have that, under some condition on the stabilizers, if we start from an orbit subspace code of dimension $k$ and maximum distance, the same action induces an optimum distance orbit full flag code with the same cardinality. This result is also true if we place as a starting code a constant dimension code of maximum distance  consisting of more than one orbit. In other words, the orbital structure, the cardinality and the property of attaining the maximum distance of  the starting subspace code of dimension $k$ will be directly translated to the flag code level. We denote by $\dot \cup$ the disjoint union of sets.

\begin{theorem}\label{cor:union_ODFC}
Let $\{ \cF^{j}=(\cF^j_1, \ldots, \cF^j_{2k-1})\}_{j=1}^m$ be a set of full flags on $\bbF_q^{2k}$ and $\bT$ a subgroup of $GL(2k,q)$ such that  the subspaces $\cF^1_k,\ldots ,\cF^m_k$ lie in different orbits of the action of $\bT$ on $\cG_q(k,2k)$ and the union $\dot\cup_{j=1}^m \orb_{\bT}(\cF^j_k)$ is a constant dimension code of maximum distance. 

If $\stab_{\bT}(\cF^j_k)\subseteq \stab_{\bT}(\cF^j_i)$, for all $i\in\{1,\ldots, 2k-1\}$ and for all $j\in\{1,\ldots, m\}$, then the union $\dot \cup_{j=1}^m \orb_{\bT}(\cF^j)$ is an optimum distance full flag code with size $\sum_{j=1}^m |\orb_{\mathbf{T}}(\cF_k^j)|$. 
\end{theorem}
\begin{proof}
 By Proposition \ref{teo:ODFCorbital}, we know that  $\orb_{\bT}(\cF^j)$ is an optimum distance full flag code, for $j=1,\ldots, m$. Moreover, since the subspaces $\cF_k^1, \ldots, \cF_k^m$ lie in different orbits,  there is no common subspace in different constant dimension codes $\orb_{\mathbf{T}}(\cF_k^j)$. Therefore, the union $\cup_{j=1}^m \orb_{\mathbf{T}}(\cF^j)$ contains exactly 
\[
\sum_{j=1}^m |\orb_{\mathbf{T}}(\cF^j) |= \sum_{j=1}^m |\orb_{\mathbf{T}}(\cF_k^j)|
\]
flags. Hence, it suffices to show that, given $A_1, \, A_2\in \bT$ and $j_1,j_2\in \{1,\ldots, m\}$ with $j_1\neq j_2$, the distance between the flags $\cF^{j_1} \cdot A_1$ and $\cF^{j_2}\cdot A_2$ is maximum.

Notice that $\cF^{j_1}_k \cdot A_1\in \orb_{\bT}(\cF^{j_1}_k)$ and $\cF^{j_2}_k \cdot A_2 \in \orb_{\bT}(\cF^{j_2}_k)$. Hence $\cF^{j_1}_k \cdot A_1\neq \cF^{j_2}_k \cdot A_2$ and $d_S(\cF^{j_1}_k\cdot A_1, \cF^{j_2}_k\cdot A_2)=2k$, by hypothesis. In other words, $\cF^{j_1}_k \cdot A_1 \cap \cF^{j_2}_k \cdot A_2= \{ 0\}$ and $\cF^{j_1}_k \cdot A_1 \oplus \cF^{j_2}_k\cdot A_2 = \bbF_q^{2k}$. Then, it holds
$$
 \cF^{j_1}_i \cdot A_1 \cap \cF^{j_2}_i \cdot A_2   \subseteq  \cF^{j_1}_k \cdot A_1 \cap \cF^{j_2}_k \cdot A_2= \{0\},  \  \mbox{ for } 1\leq i\leq k,
$$
and
$$
\bbF_q^{2k} = \cF^{j_1}_k \cdot A_1 \oplus \cF^{j_2}_k \cdot A_2  \subseteq  \cF^{j_1}_i \cdot A_1 + \cF^{j_2}_i \cdot A_2, \   \mbox{ for }  k+1\leq i\leq 2k-1. 
$$
In both cases, $\cF^{j_1}_i \cdot A_1 \cap \cF^{j_2}_i\cdot A_2$ has the minimum possible dimension. It follows that $d_S(\cF^{j_1}_i \cdot A_1 ,\cF^{j_2}_i \cdot A_2)$ is the maximum possible distance,  for all $i\in\{1,\ldots, 2k-1\}$. Therefore, $d_f(\cF^{j_1} \cdot A_1, \cF^{j_2}\cdot A_2)$ is the maximum possible distance between full flags on $\bbF_q^{2k}$, which finishes the proof.
\end{proof}

The previous result allows us to build optimum distance full flag codes on $\bbF_q^{2k}$ of higher cardinality as union of orbits under the action of the same group. Recall that, by means of Theorem \ref{teo:odfc_maxsize}, optimum distance full flag codes on $\bbF_q^{2k}$ with the maximum possible size are the ones with a $k$-spread as a projected code. Hence, given the acting group $\bT$, the number of orbits needed to get full flag codes attaining both the maximum distance and cardinality coincides  with the number of orbits in which $\bT$ partitions the planar spread we use as a starting $k$-projected code.

Coming back to the action of the group $\bG$ defined in (\ref{def: Grupo grande}), notice that if $\cF=(\cF_1, \ldots, \cF_{2k-1})$ is a full flag on $\bbF_q^{2k}$ with $\cF_k\in\cS$, where $\cS$ is the planar spread defined in (\ref{def: spread Manganiello}), the group $\bG$ satisfies the first hypothesis of Proposition \ref{teo:ODFCorbital}, since $\orb_{\bG}(\cF_k)=\cS$. Hence, we just need $\stab_{\bG}(\cF_k)$ to be a subgroup of any other $\stab_{\bG}(\cF_i)$ in order to have that $\orb_{\bG}(\cF)$ is an optimum distance full flag code. However, the next example shows that $\bG$ does not give equal stabilizers in general. 

\begin{example}\label{ex:nodisjoint}
Denote by $\cF=(\cF_1, \cF_2, \cF_3)$ the standard full flag on $\bbF_2^4$, that is, the $i$-th subspace  $\cF_i$ of $\cF$ is spanned by the first $i$ rows of the $4\times 4$ identity matrix. 
Let $P$ be the companion matrix of the primitive polynomial $p(x)=x^2+x+1 \in \bbF_2\left[ x \right]$
and consider the matrix group $\bG\leq GL(4,2)$ defined in (\ref{def: Grupo grande}).
The orbit full flag code  $\orb_\bG(\cF)$ has projected codes
$\orb_\bG(\cF)_i=\orb_\bG(\cF_i)$, for $i=1,2,3$. In particular, since $\cF_2=\rsp(I_2  \ | \ 0_2)=\cU_2$, it follows by (\ref{def: spread Manganiello orbital}) that $\orb_\bG(\cF_2)=\cS$ is the planar spread given in (\ref{def: spread Manganiello}). Moreover, $\orb_\bG(\cF_1)$ and $\orb_\bG(\cF_3)$ are collections of lines and hyperplanes of $\bbF_2^4$ respectively. Hence, all  the projected codes of $\orb_\bG(\cF)$ attain the maximum possible distance. However, $\orb_{\mathbf{G}}(\cF)$ is not an optimum distance full flag code. It suffices to consider the following matrix $A\in \bG$:
\[
A=
\begin{pmatrix}
P & 0_2 \\
0_2 & P^2
\end{pmatrix}=
\begin{pmatrix}
0_2 & I_2 \\
I_2 & 0_2
\end{pmatrix}
\begin{pmatrix}
I_2 & P^2 \\
0_2 & I_2
\end{pmatrix}
\begin{pmatrix}
0_2 & I_2 \\
I_2 & 0_2
\end{pmatrix}
\begin{pmatrix}
I_2 & P \\
0_2 & I_2
\end{pmatrix}
\begin{pmatrix}
0_2 & I_2 \\
I_2 & 0_2
\end{pmatrix}
\begin{pmatrix}
I_2 & P^2 \\
0_2 & I_2
\end{pmatrix}
\]
and the flag $\cF \cdot A \in \orb_\bG(\cF)$. One has
\[
\begin{array}{ccccc}
\cF_1  \cdot A & = & \rsp\begin{pmatrix} 
                            0 & 1 & 0 &  0 
                             \end{pmatrix}  & \neq & \cF_1,\\
\cF_2  \cdot A & = & \rsp\begin{pmatrix}
                               0 & 1 & 0 &  0 \\
                                1 & 1 & 0 & 0
                              \end{pmatrix}  & = & \cF_2,\\
\end{array}
\]
Then $A\in \stab_{\bG}(\cF_2)\setminus \stab_{\bG}(\cF_1)$. This shows that the flag code $\orb_\bG(\cF)$ is not disjoint by Proposition \ref{prop:disjorb}. Consequently $\orb_\bG(\cF)$ is not an optimum distance flag code. 
\end{example}

\begin{remark}
In the light of Corollary \ref{cor:ODFCorbital}, we see that, given a full flag $\cF=(\cF_1, \ldots, \cF_{2k-1})$ on $\bbF_q^{2k}$ with $\cF_k\in\cS$ and a subgroup $\bT$ of $\bG$, if the stabilizer of $\cF_k$ under the action of $\bT$ is trivial, then $\orb_{\bT}(\cF)$ is automatically an optimum distance flag code of size $|\bT|$. Moreover, if $|\bT|=|\cS|=q^k+1$, then $\bT$ acts regularly (that is,  transitively and providing trivial stabilizers) on $\cS$ and the full flag orbit code $\orb_{\bT}(\cF)$ has $\cS$ as its $k$-projected code. Consequently, $\orb_{\bT}(\cF)$ is an optimum distance full flag code with the best size.
Nevertheless, take into account that the action of a subgroup $\bT$ of $\bG$ on $\cF$ still could give optimum distance orbit full flag codes as long as the stabilizer subgroups of all the subspaces of the flag $\cF$ coincide (Propositions \ref{prop:contenido} and \ref{teo:ODFCorbital}). Therefore, by Theorem \ref{cor:union_ODFC}, if $|\bT|=|\cS|$ but the stabilizer subgroup of $\bT$ on $\cF_k$ is not trivial,  we will need to increase the number of orbits to get an optimum distance flag code  of size  $|\cS|$.
\end{remark}

The following section is devoted to explore the existence of subgroups of $\bG$  of order $q^k+1$  whose action on a full flag $\cF=(\cF_1, \ldots, \cF_{2k-1})$ on $\bbF_q^{2k}$ with $\cF_k\in\cS$, provides disjoint full flag codes having $\cS$ as its $k$-projected code as union of the minimum number of orbits.


\section{An orbital construction of optimum distance flag codes}\label{sec:construc}

Let us fix $n=2k$. The purpose of this section is to provide an orbital construction of optimum distance full flag codes on $\bbF_q^{2k}$ with the best size, $q^k+1$, by using a subgroup of the group $\bG$ defined in (\ref{def: Grupo grande}). To achieve this goal, we will proceed in several steps. First, we will deepen the structure of the group $\bG$ in order to find a suitable subgroup $\bH$. Secondly, this subgroup choice will give an orbital description of the Grassmannian of lines of $\bbF_{q^k}^2$ that induces an orbital construction of the planar spread $\cS$ of $\bbF_q^{2k}$ defined in (\ref{def: spread Manganiello}). Finally, by using Propositions \ref{prop:contenido}, \ref{teo:ODFCorbital}  and Theorem \ref{cor:union_ODFC} we will prove that it is possible to extend the previous construction to flags giving rise to an orbital construction of optimum distance flag codes with the best size.

\subsection{Obtaining a suitable subgroup} 

Although $\bG$ is a subgroup of $GL(2k,q)$, the field isomorphism $\phi$ given in (\ref{eq:isoalphaP}) induces the following group monomorphism $\psi$, which allows us to work on $GL(2,q^k)$:

\begin{equation}\label{eq:hom_grasmanianas}
\begin{array}{cccc}
\psi: & GL(2,q^k)                                     & \longrightarrow                     & GL(2k,q) \\
        & \begin{pmatrix}
            a & b\\
            c & d
           \end{pmatrix}
                                                               &\longmapsto            & \begin{pmatrix}
                                                                                              \phi(a) & \phi(b)\\
                                                                                               \phi(c) & \phi(d)
                                                                                                 \end{pmatrix}.
\end{array}                                                                                                 
\end{equation}
In particular, one has that  
\[
\psi\begin{pmatrix}
0 & 1\\
1 & 0
\end{pmatrix}=
\begin{pmatrix}
0 & I_k\\
I_k & 0
\end{pmatrix}
\ \text{and} \ 
\psi\begin{pmatrix}
1 & \alpha^i\\
0 & 1
\end{pmatrix}=
\begin{pmatrix}
I_k & P^i\\
0 & I_k
\end{pmatrix},
\] for $i=0, \ldots, q^k-2$. Thus, if we consider the group
\[
\bar{\bG}=
\left\langle
\begin{pmatrix}
0 & 1\\
1 & 0
\end{pmatrix}, \ 
\begin{pmatrix}
1 & \alpha^i\\
0 & 1
\end{pmatrix} \ \Big| \ i=0, \ldots, q^k-2
\right\rangle \leq GL(2,q^k),
\]
it is clear that $\psi(\bar{\bG})=\bG$ and therefore $\bar{\bG}$ and $\bG$ are isomorphic groups. 

The following statements show that $\bar{\bG}$ is extremely related to the {\em special linear group} $SL(2, q^k)$, which consists of all matrices of $GL(2,q^k)$ with determinant $1$. 

\begin{proposition}{\cite[Proposition 2, page 57]{Alperin}}\label{prop:generadores SL}
If  $\alpha$ is a generator of the multiplicative cyclic group $\bbF_{q^k}^*$, then
\[
SL(2, q^k)=
\left\langle
\begin{pmatrix}
1 & \alpha^i \\
0 & 1
\end{pmatrix}, \ 
\begin{pmatrix}
1 & 0\\
\alpha^i & 1
\end{pmatrix} \ \Big| \ i=0, \ldots, q^k-2
\right\rangle
\]
\end{proposition}

\begin{proposition}\label{prop:Estr_barG}
The group $SL(2,q^k)$ is a subgroup of $\bar{\bG}$. Moreover, if $\car(\bbF_q)=2$, then $\bar{\bG}=SL(2,q^k)$. 
In case of $\car(\bbF_q)\neq 2$, one has that $\bar{\bG}$ is the semidirect product  $\bar{\bG}=SL(2,q^k) \rtimes \left\langle\begin{pmatrix}
0 & 1\\
1 & 0
\end{pmatrix}\right\rangle$.
\end{proposition}
\begin{proof}
Observe that 
\[
\begin{pmatrix}
1 & 0\\
\alpha^i & 1
\end{pmatrix}
= 
\begin{pmatrix}
0 & 1\\
1 & 0
\end{pmatrix}
\begin{pmatrix}
1 & \alpha^i\\
0 & 1
\end{pmatrix}
\begin{pmatrix}
0 & 1\\
1 & 0
\end{pmatrix},
\]
for all $0\leq i\leq q^k-2$. Thus $SL(2,q^k)$ is a subgroup of $\bar{\bG}$ by Proposition \ref{prop:generadores SL}. On the other hand, it is clear that   
$\begin{pmatrix}
1 & \alpha^i\\
0& 1
\end{pmatrix} \in  SL(2,q^k)$, for all $0\leq i\leq q^k-2$. 
Moreover, if $\car(\bbF_q)=2$, then $\det \begin{pmatrix}
0 & 1\\
1 & 0
\end{pmatrix}=1$. Therefore $\bar{\bG}=SL(2,q^k)$ in this case.
If $\car(\bbF_q)\neq 2$, then the matrix $\begin{pmatrix}
0 & 1\\
1 & 0
\end{pmatrix}$ has multiplicative order two and it normalizes the subgroup $SL(2,q^k) \leq \bar{\bG}$. Hence $\bar{\bG}$ is the semidirect product of $\left\langle\begin{pmatrix}
0 & 1\\
1 & 0
\end{pmatrix}\right\rangle$ acting by conjugation on $SL(2,q^k)$.
\end{proof}

Now, recall the natural embedding $\varphi$ of the Grassmannian of lines of $\bbF_{q^k}^2$ into ${\cal G}_{q}(k,2k)$ defined in (\ref{eq:fieldreduction}). Since  $\varphi({\cal G}_{q^k}(1,2))={\cal S}$, the action of $\bar{\bG}$ on ${\cal G}_{q^k}(1,2)$ is equivalent to the action of $\bG$ on ${\cal S}$ in the following sense: given $\bar{A}\in \bar{\bG}$ and a line $l\in {\cal G}_{q^k}(1,2)$ given by $l=\rsp(x_1,x_2)$, it follows that
\begin{equation}\label{eq:accioequiv}
\varphi(l \cdot \bar{A})=\varphi(\rsp(x_1,x_2)\cdot \bar{A})=\varphi(\rsp(x_1,x_2))\cdot \psi(\bar{A})=\varphi(l)\cdot \psi(\bar{A}).
\end{equation}
This way, we can work with the action of $\bar{\bG}$ (or subgroups of $\bar{\bG}$) on ${\cal G}_{q^k}(1,2)$ by simplicity and then easily translate it to the action of $\bG$ (or subgroups of $\bG$) on ${\cal S}$ via $\varphi$. In particular,   $\varphi(\rsp(1, 0))={\cal U}_k$ and using (\ref{eq:accioequiv}) we obtain
\begin{eqnarray}\label{eq:Svarphi}
{\cal S}= \orb_{\bG}({\cal U}_k)  & = & \{{\cal U}_k\cdot A\ |\ A\in \bG\} \nonumber \\
                                               & = &  \{\varphi(\rsp(1, 0))\cdot  \psi(\bar{A})\ |\ \bar{A}\in \bar{\bG} \} \nonumber \\
                                               & = & \{\varphi(\rsp(1, 0)\cdot \bar{A})\ |\ \bar{A}\in \bar{\bG} \} \nonumber \\
                                               & = & \varphi(\orb_{\bar{\bG}}(\rsp(1, 0))). \nonumber
\end{eqnarray}
Consequently, 
\begin{equation}\label{eq:orbrectes}
{\cal G}_{q^k}(1,2)=\orb_{\bar{\bG}}(\rsp(1, 0))
\end{equation} 
since $\varphi$ is injective. Thus, given that $\bar{\bG}$ acts transitively on ${\cal G}_{q^k}(1,2)$, we wonder if there is any subgroup of $\bar{\bG}$ acting regularly on ${\cal G}_{q^k}(1,2)$. To this end, we study first the existence of subgroups of order  $q^k+1$ in $\bar{\bG}$.

\begin{remark}
Observe that we can restrict our study to subgroups of $SL(2,q^k)$. In case the characteristic is even, it is clear by Proposition \ref{prop:Estr_barG}. On the other hand, for odd characteristic, notice that the matrix $\begin{pmatrix}
0 & 1\\
-1 & 0
\end{pmatrix}$ has determinant equal to one. Hence, 
\begin{eqnarray}
\cG_{q^k}(1,2) & = & \{\rsp(0,1),\rsp(1,0),\rsp(1,\alpha^i) \ | \  i=1,\ldots q^k-2\} \nonumber \\
                         & = &  \orb_{SL(2,q^k)}(\rsp(1, 0))\nonumber                      
\end{eqnarray}
in both cases.
\end{remark}

At this point, we take into account the existence of {\em Singer subgroups} of $SL(2,q^k)$ whose order is precisely $q^k+1$. Singer subgroups of classical linear groups are well-known in finite group theory (see \cite{Hes, HupI} for instance).  In general, a Singer subgroup of $GL(n,q)$ is a cyclic subgroup of order $q^{n}-1$ generated by a {\em Singer cycle} or {\em cyclic projectivity} (see \cite{Hirschfeld98}). In our context and from \cite{HupI}, Singer subgroups of $GL(2,q^k)$ can be constructed in the following way. Let us consider a primitive element $\omega\in \bbF_{q^{2k}}$ and the companion matrix $M \in GL(2, q^k)$ of the minimal polynomial of $\omega$ over $\bbF_{q^k}$. The matrix $M$ is called a {\em Singer cycle} of $GL(2,q^k)$ and any conjugate subgroup of $\langle M \rangle$ is said to be a {\em Singer subgroup} of $GL(2,q^k)$. Recall from Section \ref{subsec:finitefields} that the multiplicative order of $M$ is equal to $q^{2k}-1$ and that we have the field isomorphism 
\[
\bbF_{q^{2k}}\cong \bbF_{q^k}\left[M\right]=\{0\} \cup \langle M\rangle.
\]
As a consequence, the group $\langle M\rangle$ acts regularly by multiplication on the elements of $\bbF_{q^{2k}}^*\cong \langle M\rangle$, which can be identified with the set of non-zero vectors of $\bbF_{q^k}^2$.

The intersection of a Singer subgroup of $GL(2,q^k)$ with $SL(2,q^k)$ is called  a {\em Singer subgroup} of $SL(2,q^k)$. As it is stated in Section \ref{subsec:finitefields}, the eigenvalues of $M$ are $\omega$ and $\omega^{q^k}$ and $\det(M)=\omega^{q^k+1}$. Therefore, $\det(M^{q^k-1})=\det(M)^{q^k-1}=1$, since the multiplicative order of $\omega$ is $q^{2k}-1$. Consequently, the group 
\begin{equation}\label{eq:def_barH}
\bar{\bH}=\langle M\rangle\cap SL(2,q^k)=\langle M^{q^k-1}\rangle 
\end{equation}
is a Singer subgroup of $SL(2,q^k)$ and its order is $q^k+1$. From now on, we fix this group $\bar{\bH}$. Throughout the following section, we will study how  $\bar{\bH}$ acts on the set of lines of  $\bbF_{q^k}^2$.

\subsection{An orbital description of the spread of lines}

As it was stated in (\ref{eq:orbrectes}), we know that ${\cal G}_{q^k}(1,2)=\orb_{\bar{\bG}}(\rsp(1, 0))$. Since $|{\cal G}_{q^k}(1,2)|=|\bar{\bH}|=q^k+1$, it seems quite natural to wonder if $\bar{\bH}$ acts regularly on ${\cal G}_{q^k}(1,2)$. As we will see, this is true for even characteristic whereas for odd characteristic this action is not transitive and we obtain $\cG_{q^k}(1,2)$ as the union of exactly two orbits.

To do this, let us come back to the group $\langle M\rangle \subset GL(2, q^k)$ for a while and study its action on $\cG_{q^k}(1,2)$.
Next we see that the set of matrices in $\langle M\rangle$ stabilizing an arbitrary line is exactly the group of scalar matrices. 

\begin{lemma}\label{lem:stabM}
Let $l$ be an arbitrary line of ${\cal G}_{q^k}(1,2)$. It holds that 
\[
\stab_{\langle M\rangle}(l)=\{a I_2\ |\ a\in \bbF_{q^k}^*\}.
\]
\end{lemma}
\begin{proof}
Let $A\in \stab_{\langle M\rangle}(l)$ and consider a non-zero vector $(x_1, x_2)\in \bbF_{q^k}^2$ such that $l=\rsp(x_1, x_2)$. Thus, $(x_1, x_2) A=a(x_1, x_2)$, for some $a\in \bbF_{q^k}^*$. Equivalently, 
$(x_1, x_2) (A -a I_2)=(0, 0)$. But $A -a I_2\in \bbF_{q^k}\left[M\right]$, which is a field, and thus the only non-invertible matrix is the null matrix. It follows that $A=aI_2$ and the result is true since the other inclusion is trivial.
\end{proof}

From the definition of $\bar{\bH}$ (see (\ref{eq:def_barH})), one has that $\stab_{\bar{\bH}}(l)=\stab_{\langle M\rangle}(l)\cap SL(2,q^k)$ for any arbitrary line $l$ in ${\cal G}_{q^k}(1,2)$.  Consequently, the next result follows directly from the previous lemma.

\begin{lemma}\label{lem:stabHbar}
Let $l$ be an arbitrary line of ${\cal G}_{q^k}(1,2)$. It holds that 
\[
\stab_{\bar{\bH}}(l)= \{a I_2\ |\ a\in \bbF_{q^k}^*,\ a^2=1\}=\{I_2,-I_2\}.
\]
\end{lemma} 

\begin{remark}
Note that the stabilizer subgroups of $\bar{\bH}$ on all the lines in $\cG_{q^k}(1,2)$ coincide. Moreover, if the field characteristic is even, this stabilizer is the trivial group. On the other hand, for odd characteristic, it consists of two elements, each of one fixing the whole set $\cG_{q^k}(1,2)$.
\end{remark}

\begin{proposition}\label{prop:AccioHbarra}
If  $\car(\bbF_q)=2$, then the group $\bar{\bH}$ acts regularly on ${\cal G}_{q^k}(1,2)$.
If  $\car(\bbF_q)\neq 2$, there are exactly two orbits of the action of $\bar{\bH}$ on $\cG_{q^k}(1,2)$, each of them contains $\frac{q^k+1}{2}$ elements.
\end{proposition}
\begin{proof}
If the field characteristic is even, the result is clear. Assume that $\car(\bbF_q)\neq 2$. From Lemma \ref{lem:stabHbar}, we  have that $\stab_{\bar{\bH}}(l)=\langle -I_2\rangle$, for any line $l\in \cG_{q^k}(1,2)$. Hence, $\orb_{\bar{\bH}}(l)$ contains exactly $\frac{q^k+1}{2}$ lines and $\cG_{q^k}(1,2)$ is partitioned into two orbits with the same cardinality. 
\end{proof}

In view of the previous result, the Grassmannian $\cG_{q^k}(1,2)$ can be described as a unique orbit in even characteristic. For odd characteristic, it is enough to consider a couple of lines lying in different orbits to describe $\cG_{q^k}(1,2)$ as the union of their respective orbits. 
As an example of this construction, the next result gives a specific partition of $\cG_{q^k}(1,2)$ in odd characteristic. 

\begin{proposition}\label{prop:eximpar}
If  $\car(\bbF_q)\neq 2$, then 
\[
\cG_{q^k}(1,2)=\orb_{\bar{\bH}}(\rsp(1,0))\ \dot\cup\ \orb_{\bar{\bH}}(\rsp(0,1)).
\] 
\end{proposition}
\begin{proof}
By means of Proposition \ref{prop:AccioHbarra}, to prove the statement, it suffices to show that $\orb_{\bar{\bH}}(\rsp(1,0))$ and $\orb_{\bar{\bH}}(\rsp(0,1))$ are different orbits. 

Arguing by contradiction, assume that $\rsp(0,1)\in \orb_{\bar{\bH}}(\rsp(1,0))$. Thus, there exist $A\in\bar{\bH}$ such that $\rsp(0,1)=\rsp(1,0)\cdot A$, that is, the first row of $A$ must be of the form $(0,a)$ for some $a\in \bbF_{q^k}^\ast$. 
On the other hand, recall that the first row of the companion matrix $M$ is $(0,1)$. It follows that the first row of $A-aM$ is the zero vector. But $A-aM\in  \bbF_{q^k}\left[M\right] $, which is a field. Thus the null matrix is the only non-invertible matrix on it. Therefore, $A=aM$ and, in particular, $A^{\frac{q^{2k}-1}{2}}=(aM)^{\frac{q^{2k}-1}{2}}$. Now, since $q^k-1$ is an even number, we can put
\[
A^{\frac{q^{2k}-1}{2}}=(A^{q^{k}+1})^{\frac{q^{k}-1}{2}}=I_2,
\]
since $A\in \bar{\bH}$ and $|\bar{\bH}|=q^k+1$. In contrast, 
\[
(aM)^{\frac{q^{2k}-1}{2}}=a^{\frac{q^{2k}-1}{2}} M^{\frac{q^{2k}-1}{2}}=(a^{q^k-1})^{\frac{q^{k}+1}{2}} M^{\frac{q^{2k}-1}{2}}=M^{\frac{q^{2k}-1}{2}}
\]
since $a$ is a element of the group $\bbF_{q^k}^\ast$. We conclude that $M^{\frac{q^{2k}-1}{2}}=I_2$, which is a contradiction since the order of $M$ is $q^{2k}-1$. 
\end{proof}

\begin{remark}\label{rem:impar_no_regular}
At this point, a natural question is if, for odd characteristic, there exist any other subgroup of order $q^k+1$ of $\bar{\bG}$ providing the spread of lines of $\bbF_{q^k}^2$ as the unique orbit of a regular action. The answer to this question is that it is not true in general. For instance,  if $q=3$ and $k=2$, we have checked by using GAP, that there is no subgroup of order $q^k+1=10$ of $\bar{\bG}\cong SL(2, 9)\rtimes C_2$ acting transitively on $\cG_{9}(1,2)$ (by $C_2$ we denote the cyclic group of order $2$).
\end{remark}


\subsection{From the spread of lines to  the planar spread ${\cS}$}\label{sec:planarorb}
Now, we will translate the orbital structure of $\cG_{q^k}(1,2)$ under the action of $\bar{\bH}$, studied in the previous section, to the planar spread ${\cal S}$ of $\bbF_q^{2k}$ defined in $(\ref{def: spread Manganiello})$. To do so, we will use the field reduction map $\varphi$ defined in (\ref{eq:fieldreduction}) and the group monomorphism $\psi$ defined in (\ref{eq:hom_grasmanianas}).  Let us consider $\bH=\psi(\bar{\bH})$, the subgroup of $\bG$ isomorphic to $\bar{\bH}$ via $\psi$.

With this notation, the following lemma states that all subspaces in $\cS$ have the same stabilizer subgroup under the action of  $\bH$.

\begin{lemma}\label{lem:stabH}
Let $\cL$ be an arbitrary subspace in ${\cal S}$. It holds that 
\[
\stab_{\bH}(\cL)=\{I_{2k}, -I_{2k}\}.
\]
\end{lemma}
\begin{proof}
Since ${\cal S}=\varphi(\cG_{q^k}(1,2))$, there exists a line $l \in \cG_{q^k}(1,2)$ such that $\varphi(l)={\cal L}$. Now, using the equivalence of the action provided in (\ref{eq:accioequiv}), it is clear that $\stab_{\bH}(\cL)=\psi(\stab_{\bar{\bH}}(l))$. Thus, the result follows  from Lemma \ref{lem:stabHbar}.
\end{proof}

Now, the orbital structure of the planar spread ${\cal S}$ under the action of $\bH$ can be described straightforwardly:

\begin{theorem}\label{teo:accioHsobreS}
If $\car(\bbF_q)=2$, then $\bH$ acts regularly on ${\cal S}$.
If $\car(\bbF_q)\neq 2$, then there are exactly two orbits of the action of $\bH$ on $\cS$, each of them with the same cardinality $\frac{q^k+1}{2}$.
\end{theorem}
\begin{proof}
Assume that $\car(\bbF_q)=2$ and consider an arbitrary subspace ${\cal L}\in{\cal S}$. Since ${\cal S}=\varphi(\cG_{q^k}(1,2))$, there exists a line $l \in \cG_{q^k}(1,2)$ such that $\varphi(l)={\cal L}$. Now, Proposition \ref{prop:AccioHbarra} states that $\cG_{q^k}(1,2)=\orb_{\bar{\bH}}(l)$ and we can use (\ref{eq:accioequiv}) to see that
\[
\cS=\varphi(\cG_{q^k}(1,2))=\varphi(\orb_{\bar{H}}(l))=\{\varphi(l \cdot \bar{A})\ |\ \bar{A}\in \bar{\bH} \}= 
\{{\cal L}\cdot A\ |\ A\in \bH \}=\orb_{\bH}(\cal L) .
\]
Suppose now that $\car(\bbF_q)\neq 2$. By Proposition \ref{prop:AccioHbarra}, we know that 
\[
\cG_{q^k}(1,2)=\orb_{\bar{\bH}}(l_1) \ \dot\cup \ \orb_{\bar{\bH}}(l_2),
\]
where $l_1,l_2\in \cG_{q^k}(1,2)$ lie in different orbits under the action of $\bar{\bH}$. Let ${\cal L}_1, {\cal L}_2 \in \cS$ such that $\varphi(l_1)={\cal L}_1$ and $\varphi(l_2)={\cal L}_2$. Then, we have that 
\[
\varphi(\orb_{\bar{\bH}}(l_i))=\{\varphi(l_i \cdot \bar{A})\ |\ \bar{A}\in \bar{\bH} \}= 
\{{\cal L}_i \cdot A\ |\ A\in \bH \}=\orb_{\bH}({\cal L}_i),
\]
for $i=1,2$. In particular, $\orb_{\bH}({\cal L}_i)$ has $\frac{q^k+1}{2}$ elements since $\varphi$ is injective. Moreover, 
\begin{eqnarray}
{\cal S}=\varphi(\cG_{q^k}(1,2)) & = & \varphi(\orb_{\bar{\bH}}(l_1) \ \dot\cup \ \orb_{\bar{\bH}}(l_2))\nonumber \\
                                                  & = & \varphi(\orb_{\bar{\bH}}(l_1) )\ \dot\cup \  \varphi(\orb_{\bar{\bH}}(l_2))=
\orb_{\bH}({\cal L}_1)\  \dot\cup \   \orb_{\bH}({\cal L}_2). \nonumber
\end{eqnarray}
\end{proof}

Recall that ${\cal U}_k=\rsp( I_k | 0_k )\in \cG_{q}(k,2k)$. Now we denote ${\cal V}_k=\rsp( 0_k | I_k )\in \cG_{q}(k,2k)$. Since $\varphi(\rsp(1,0))={\cal U}_k$ and  $\varphi(\rsp(0,1))={\cal V}_k$,  by (\ref{eq:accioequiv}) and  Proposition \ref{prop:eximpar}, in odd characteristic we have a nice orbital description of the spread $\cS$.

\begin{proposition}\label{prop:S_impar}
If $\car(\bbF_q)\neq 2$, then  ${\cal S}=\orb_{\bH}({\cal U}_k)\ \dot\cup\  \orb_{\bH}({\cal V}_k)$.
\end{proposition}

\begin{remark}
By Remark \ref{rem:impar_no_regular} and using the equivalent action provided by (\ref{eq:accioequiv}), for odd characteristic,  it is not possible, in general, to obtain $\cal S$ as the only orbit of the regular action of a subgroup (of order $q^k+1$) of $\bG$.
\end{remark}

\subsection{From the planar spread $\cS$ to flags}\label{subsec: ODFC}

In this section we will focus on how to extend the action of $\bH$ on $k$-dimensional subspaces given in Section \ref{sec:planarorb} to full flags on $\bbF_q^{2k}$,  in order to obtain flag codes with the maximum distance having the spread ${\cal S}$ defined in (\ref{def: spread Manganiello}) as its $k$-projected code. The following results show that $\bH$ always provides optimum distance full flag codes  of the best size as the union of at most two orbits.

\begin{theorem}\label{cor:size_ODFC}
Let $\cF=(\cF_1, \ldots, \cF_{2k-1})$ be a full flag on $\bbF_q^{2k}$ such that $\cF_k\in {\cal S}$. 

If $\car(\bbF_q)=2$, then $\orb_{\bH}(\cF)$ is an optimum distance full flag code with the largest possible size, which is $q^k+1$.

If $\car(\bbF_q)\neq 2$, then $\orb_{\bH}(\cF)$ is an optimum distance full flag code with size $\frac{q^k+1}{2}$.
\end{theorem}
\begin{proof}
Since  $\cF_k\in {\cal S}$, it is clear that $\orb_{\bH}(\cF_k)$ has maximum distance. Moreover, Lemma \ref{lem:stabH} states that $\stab_{\bH}(\cF_k)$ stabilizes $\cF_i$ for all $i=1, \ldots, 2k-1$. Therefore, Proposition \ref{teo:ODFCorbital} allows us to conclude that $\orb_{\bH}(\cF)$ is an optimum distance full flag code of cardinality $|\orb_{\bH}(\cF)|=|\orb_{\bH}(\cF_k)|$. Finally, this cardinality is given by  Theorem \ref{teo:accioHsobreS}.
\end{proof}

For odd characteristic, we can give an optimum distance full flag code of the largest possible size by considering the union of two orbit flag codes constructed as in Theorem \ref{cor:size_ODFC} and using Theorem \ref{cor:union_ODFC}.

\begin{proposition}\label{prop:largest_size_impar}
Assume that $\car(\bbF_q)\neq 2$. 
Let $\cF=(\cF_1, \ldots, \cF_{2k-1})$  and $\cF'=(\cF'_1, \ldots, \cF'_{2k-1})$ be two full flags on $\bbF_q^{2k}$ such that $\cF_k,\,  \cF'_k\in {\cal S}$ lie in different orbits of the action of $\bH$ on $\cS$. Then the full flag code 
\[
\cC=\orb_{\bH}(\cF)\ \dot\cup\  \orb_{\bH}(\cF')
\]
is  an optimum distance flag code with the largest possible size, that is, $q^k+1$. 
\end{proposition}
\begin{proof}
Notice that, by Lemma \ref{lem:stabH}, we can use Theorem \ref{cor:union_ODFC} to conclude that $\cC$ is
an optimum distance full flag code of size  $|\orb_{\bH}(\cF_k)|+|\orb_{\bH}(\cF'_k)|$, which is $q^k+1$ by Theorem \ref{teo:accioHsobreS}.
\end{proof}

As a consequence, we can directly translate the orbital description of the spread $\cS$ given in Proposition \ref{prop:S_impar} to obtain a specific construction of an optimum distance full flag code with the largest possible size on fields of odd characteristic:

\begin{corollary}
Assume that $\car(\bbF_q)\neq 2$. 
Let $\cF=(\cF_1, \ldots, \cF_{2k-1})$  and $\cF'=(\cF'_1, \ldots, \cF'_{2k-1})$ be two full flags on $\bbF_q^{2k}$ such that $\cF_k=\cU_k$ and $\cF'_k=\cV_k$. Then the full flag code 
\[
\cC=\orb_{\bH}(\cF)\ \dot\cup\  \orb_{\bH}(\cF')
\]
is an optimum distance flag code with the largest possible size, that is, $q^k+1$. 
\end{corollary}

\begin{remark}
The main goal of this paper was to provide an orbital construction of optimum distance full flag codes on $\bbF_q^{2k}$ having the spread $\cS$ as its $k$-projected code by using some subgroup of the group $\bG$. By Theorem \ref{teo:odfc_maxsize}, these flag codes have both maximum distance and cardinality. For the case, $q=3$ and $k=2$, we have checked by using GAP that no subgroup $\bN$ of $\bG$ such that  ${\cal S}=\orb_{\bN}({\cal U}_2)$ is able to extend its action to the standard full flag on $\bbF_3^{4}$ providing an optimum distance flag code as a single orbit. Hence, for odd characteristic, it is not possible to give a general construction of orbit full flag code, under the action of a subgroup of $\bG$, with the maximum distance and the best size for that distance. Consequently, the action of our subgroup $\bH$ allows us to achieve our objective with the lowest possible number of orbits.

\end{remark}

\section{Conclusions}
In this paper we have adressed the construction of optimum distance full flag codes having an orbital structure. To do this, we have stated first some useful properties of orbit flag codes under the action of an arbitrary subgroup of $GL(n,q)$. For the particular case $n=2k$, we have given sufficient conditions to obtain optimum distance orbit full flag codes on $\bbF_q^{2k}$, starting from orbit codes of dimension $k$ and maximum distance. 

Moreover, we have shown how, under these conditions, these starting orbit codes completely drive our construction in the following sense: their property of being of maximum distance, as well as their cardinality and orbital structure, are perfectly translated to the flag codes level. As a consequence, the cardinality of optimum distance full flag codes constructed in this way can be increased by considering a starting subspace code of dimension $k$ consisting of more than one orbit. Besides, we have presented a specific orbital construction of optimum distance full flag codes on $\bbF_q^{2k}$ with the best size using at most two orbits.

\end{document}